\documentclass{article}
\usepackage{ijcai13}
\usepackage{times}
\usepackage{amsmath}
\usepackage{amssymb}
\usepackage{amsthm}
\usepackage{extarrows}
\usepackage{arcs}
\usepackage{bbm}
\usepackage{subfigure}
\usepackage{epsfig}
\usepackage{url}
\usepackage{graphics}

\begin{document}

\title{A Game-theoretic Machine Learning Approach for Revenue Maximization in Sponsored Search}
\author{Di He$^\dag$, Wei Chen$^\ddag$, Liwei Wang$^\dag$ \and Tie-Yan Liu$^\ddag$ \\
$^\dag$Key Laboratory of Machine Perception, MOE, School of Electronics Engineering and Computer Science,\\
Peking University, Beijing, P.R.China\\
$^\ddag$Microsoft Research Asia, Building 2, No. 5 Danling Street, Beijing, P.R.China \\
$^\dag$\{hedi,wanglw\}@cis.pku.edu.cn, $^\ddag$\{wche,tyliu\}@microsoft.com}

\maketitle
\newtheorem{definition}{Definition}
\newtheorem{theorem}{Theorem}
\newtheorem{lemma}{Lemma}
\newtheorem{corollary}{Corollary}
\newtheorem{proposition}{Proposition}
\newtheorem{conjecture}{Conjecture}
\begin{abstract}
Sponsored search is an important monetization channel for search engines, in which an auction mechanism is used to select the ads shown to users and determine the prices charged from advertisers. There have been several pieces of work in the literature that investigate how to design an auction mechanism in order to optimize the revenue of the search engine. However, due to some unrealistic assumptions used, the practical values of these studies are not very clear. In this paper, we propose a novel \emph{game-theoretic machine learning} approach, which naturally combines machine learning and game theory, and learns the auction mechanism using a bilevel optimization framework. In particular, we first learn a Markov model from historical data to describe how advertisers change their bids in response to an auction mechanism, and then for any given auction mechanism, we use the learnt model to predict its corresponding future bid sequences. Next we learn the auction mechanism through empirical revenue maximization on the predicted bid sequences. We show that the empirical revenue will converge when the prediction period approaches infinity, and a Genetic Programming algorithm can effectively optimize this empirical revenue. Our experiments indicate that the proposed approach is able to produce a much more effective auction mechanism than several baselines.
\end{abstract}

\section{Introduction}

Sponsored search is an important means of Internet monetization, and is the driving force of major search engines today. In sponsored search, keyword auction is used to determine the ranking and pricing of the ads, and thus affects the revenue of the search engine. The research on auction mechanism design has attracted the attention of many researchers from the areas of artificial intelligence and electronic commerce \cite{broder2007semantic,pclick,GSP}.

In keyword auctions, each advertiser is required to submit a bid for his/her ad. When the keyword is issued by a web user, the search engine will rank the ads, show the top-ranked ones to the user, and then charge the corresponding advertisers in certain conditions (e.g., if their ads are clicked). Generalized second price (GSP) auction is a family of auction mechanisms that has been popularly used by today's search engines, which ranks ads according to the products of their bid prices and quality scores; and charges a clicked ad by the minimum bid price to maintain its current rank position.

Given the critical role of the auction mechanism in sponsored search, many people have studied the optimization of its performance in terms of search engine revenue \cite{lahaie2007revenue,garg2007design,zhu2009optimizing,radlinski2008optimizing,zhu2009revenue}. These studies can be categorized into two groups. The first group \cite{lahaie2007revenue,garg2007design} addresses the problem from a game-theoretic perspective. Some works in this group optimize the worse-case revenue in symmetric Nash equilibria; some others consider Bayesian optimal mechanism design with the knowledge of bidders' value distribution. They usually require strong (and somewhat unrealistic) assumptions on the game, such as the full information about the values/bids of the bidders (or their distributions) and the full rationality of the bidders. However, empirical evidence has shown that these assumptions might not hold in practice, especially considering that sponsored search auctions are complex, fast-paced, and involve a very large number of bidders\cite{pin2011stochastic,duong2011discrete,edelman2007strategic}. In such a real situation, advertisers will have limited information access and bounded rationality, therefore it is inappropriate to perform classical game-theoretic analysis based on the aforementioned assumptions. The second group \cite{zhu2009optimizing,radlinski2008optimizing,zhu2009revenue} tackles the task from a machine learning perspective, in particular, directly adopts conventional machine learning methods to optimize the revenue on historical bidding data. The underlying assumption is that the bid distribution will not change over time thus the optimal auction mechanism learnt from historical data can generalize to future test data. However, as we know, there is a so-called \emph{second-order effect} in sponsored search, i.e., many advertisers will adjust the bids according to their utilities, as the responses to a new auction mechanism. As a result, the ``optimal'' auction mechanism directly learnt from historical bidding data will no longer be optimal after it is deployed since advertisers will change their bids in response to it.

To overcome the above drawbacks, we propose a novel approach, which can naturally combine game theory and machine learning using a bilevel optimization framework, so as to simultaneously avoid the strong assumptions and handle the second-order effect. For ease of reference, we call the approach a \emph{game-theoretic machine learning} approach. Generally speaking, in the approach, we first learn an advertiser behavior model from historical data to describe how advertisers change their bids in response to a given auction mechanism. We then predict advertiser's bid changes using this model during the process of auction mechanism learning.

Specifically, when building the advertiser behavior model, we consider that the advertisers do not have a very good knowledge about the detailed auction mechanism or the bids of other advertisers. Instead, they adjust their bids mainly based on the previous bids of their own and the signals provided by the search engine. We assume the time-homogeneous Markov property for this dependency, and the signals mainly include the number of impressions, the number of clicks, and the average cost per click (which are referred to as Key Performance Index, or KPI for short). We learn the parametric transition matrix of the Markov model by a maximum likelihood method. For the next step, the learnt advertiser behavior model is used to predict the future bid sequences of the advertisers after a new auction mechanism is deployed. Then we estimate the revenue of the search engine over the predicted mechanism-dependent bids in a finite time period (based on the historical logs), and call this quantity the \emph{empirical revenue} under the advertiser behavior model. We prove that the empirical revenue of an arbitrary auction mechanism will converge when the length of the sequence approaches infinity, and then apply Genetic Programming to search for the optimal auction mechanism in terms of this empirical revenue.

We have conducted preliminary experiments to test the proposed approach. The experiments show very promising results: our learnt auction mechanism can outperform several baselines including the widely used classic GSP auction.

\section{Game-theoretic Machine Learning Approach}

In this section, we introduce our proposed approach for revenue maximization in sponsored search, which we call a \emph{game-theoretic machine learning} approach\footnote{In the literature of online learning, there is a branch called \emph{game-theoretic learning} or \emph{adaptive mechanism design} \cite{pardoe2006adaptive}. The main difference between our work and theirs is that we learn a mechanism based on historical data, while they optimize the performance of the mechanism using an online scheme.}.

\subsection{Preliminaries}
According to \cite{caragiannis2012efficiency,pclick,GSP,varian2007position}, sponsored search auctions can be mathematically formulated as follows. Assume there are $m$ ads from $m$ advertisers; each advertiser $i$ has a value per click $v_i$ on his/her ad; and each ad $i$ is represented by a feature vector $x_i\in\mathcal{R}^d$.

Before the auction starts, the search engine receives a bid vector $\mathbf{b}=(b_1, \cdots ,b_i, \cdots ,b_m)$ from the advertisers where the $i$-th component $b_i$ is the bid of the $i$-th advertiser. When the keyword is issued by a web user $u$, the search engine will rank the ads according to the products of their quality scores $f(x_i)$ and bid prices $b_i$, and show the top-ranked ads to the user. As a common practice, the quality score $f(x_i)$ is defined as a compound function $f(x_i)=g(\widehat{CTR} (x_i))$, where $\widehat{CTR}(x_i)$ predicts the click probability of ad $i$ at the top-1 position, and $g(.)$ is a monotone function. In the literature, different $g(.)$ functions have been used. For example, $g(t)=1$ was used by Yahoo! in early 2000s; later $g(t)=t$ has been used by several search engines; and recently $g(t)=t^{\alpha}$ has been considered \cite{lahaie2011efficient,lahaie2007revenue}.

If ad $i$ is placed at position $j$ and is clicked by the user, advertiser $i$ will be charged by the following amount of money according to the Generalized Second Price rule,
{\small\begin{eqnarray}
Pay_{f,\mathbf{b}}(i)&=&\frac{f(x_{\sigma_{f,\mathbf{b}}(j+1)})b_{\sigma_{f,\mathbf{b}}(j+1)}}{f(x_{i})}\nonumber
\end{eqnarray}}
where $\sigma_{f,\mathbf{b}}(j)$ is the index of the ad ranked at position $j$ according to  $f(x)b$.

The corresponding utility function for advertiser $i$ can be written as follows,
{\small\begin{eqnarray}
Uti_{f,\mathbf{b}}(i)&=&[v_i-Pay_{f,\mathbf{b}}(i)]\mathbbm{1}_{\{\text{user } u \text{ clicks on ad } i\}}\nonumber
\end{eqnarray}}
where $\mathbbm{1}_{\{.\}}$ is the indicator function.

The revenue of the search engine can be obtained as below,
{\small\begin{eqnarray}
&&r(\mathbf{b},f,u) =\nonumber\\
&&\sum_iPay_{f,\mathbf{b}}(\sigma_{f,\mathbf{b}}(i))\mathbbm{1}_{\{\text{user } u \text{ clicks on the ad at position } i\}}
\end{eqnarray}}
Considering that in GSP, both the ranking and pricing rules are determined once the quality score $f$ is given, we will also refer to $f$ as the auction mechanism if there is no confusion in the context.

\subsection{Game Theory vs. Machine Learning}

Given an auction mechanism space $\mathcal{F}$, one of the goals of the search engine is to design an auction mechanism $f^*\in\mathcal{F}$ to maximize its revenue. In the literature of game theory, there have been some attempts along this direction. For example, in \cite{lahaie2007revenue} the worst-case revenue in the symmetric Nash equilibria is maximized, and in \cite{garg2009optimal,garg2007design} the Bayesian optimal auction mechanism design is investigated with the value distribution of the bidders as public knowledge. In these works, some ideal assumptions have been employed. For instance, one usually assumes that the values/bids (or their distributions) of the advertisers as well as the auction mechanism of the search engine are publicly accessible; however, in reality, an advertiser can only see the information about his/her own ad (e.g., the bids and KPI associated with the ad in a certain period of time) but have no access to the information of other advertisers or the search engine. For another instance, it is usually assumed that advertisers are rational enough to take the best responses so as to maximize their utilities; however, in reality, advertisers have very different and diverse bidding behaviors: some advertisers are very active while some others seldom change their bids; some are highly capable of placing near-optimal bids while some others are not. As a result, the practical values of the aforementioned game-theoretic attempts are not very clear.

In recent years, some machine learning researchers have tried to optimize the empirical revenue of an auction mechanism based on historical bidding data, in order to avoid the use of unrealistic game-theoretic assumptions. For example, in \cite{zhu2009optimizing,zhu2009revenue},the authors propose to simultaneously optimize the revenue and relevance of the auction mechanism on historical bidding data. However, these works actually introduce another kind of assumption, which is about the i.i.d. distribution of bidding data. This is a very foundational assumption used by statistical machine learning, which guarantees that the model learnt from the training set can generalize to the future test set. Unfortunately, this assumption does not hold either in the context of sponsored search. As mentioned in the introduction, the so-called second-order effect makes the optimal auction mechanism on the historical bids no longer optimal since advertisers will change their bids in response to the auction mechanism in the future.

\subsection{Bilevel Optimization Framework}

We propose a bilevel optimization framework to address the aforementioned problems with previous works. The framework contains two levels of optimization, which models the advertiser's responses in the inner level and optimizes the revenue of the search engine in the outer level. The framework naturally combines game theory and machine learning, and therefore we call it a \emph{game-theoretic machine learning} approach.

The proposed framework can be intuitively depicted using Figure 1, and can also be mathematically characterized as follows.
{\small\begin{eqnarray}
\max_{f\in\mathcal{F}}R(f,g,S)\\
s.t. \min_{g\in\mathcal{G}} L(g,S)
\end{eqnarray}}
where $S$ represents the historical data; $\mathcal{G}$ is the space of advertiser behavior models; $L(g,S)$ is the loss function to facilitate the learning of the behavior model; $\mathcal{F}$ is the space of the auction mechanisms; and $R(f,g,S)$ is the empirical revenue defined based on the historical data, the advertiser behavior model, and the auction mechanism $f$.

As can be seen from the figure, two kinds of training data are used: advertisers' auction logs and web users' query and click-through logs. The former record the historical bids of the advertisers in a period of time; the latter contain the keywords issued by users and their click behaviors in the same period. As indicated by the formulas, the inner level of our framework is to learn a mechanism-dependent advertiser behavior model from historical auction logs, which is used to predict advertisers' bid changes in the future. Since the advertiser behavior model is mechanism-dependent, we can use the model to predict advertisers' future bid sequences based on the users' query and click-through data with respect to different auction mechanisms. Then the outer level is to learn the optimal auction mechanism in terms of the empirical revenue on the predicted bid sequences.

Here we would like to emphasize that it is actually very fundamental to model sponsored search using bilevel optimization (and we regard this as one of our contributions). The bilevel optimization framework cannot only characterize our proposed approach, but also contain many previous works as its special cases. For example, the framework can cover the worst-case optimal mechanism design \cite{lahaie2007revenue} if the inner level optimization problem characterizes a set of equilibrium conditions on the bids.

\begin{figure}
\centering
\includegraphics[width=0.38\textwidth]{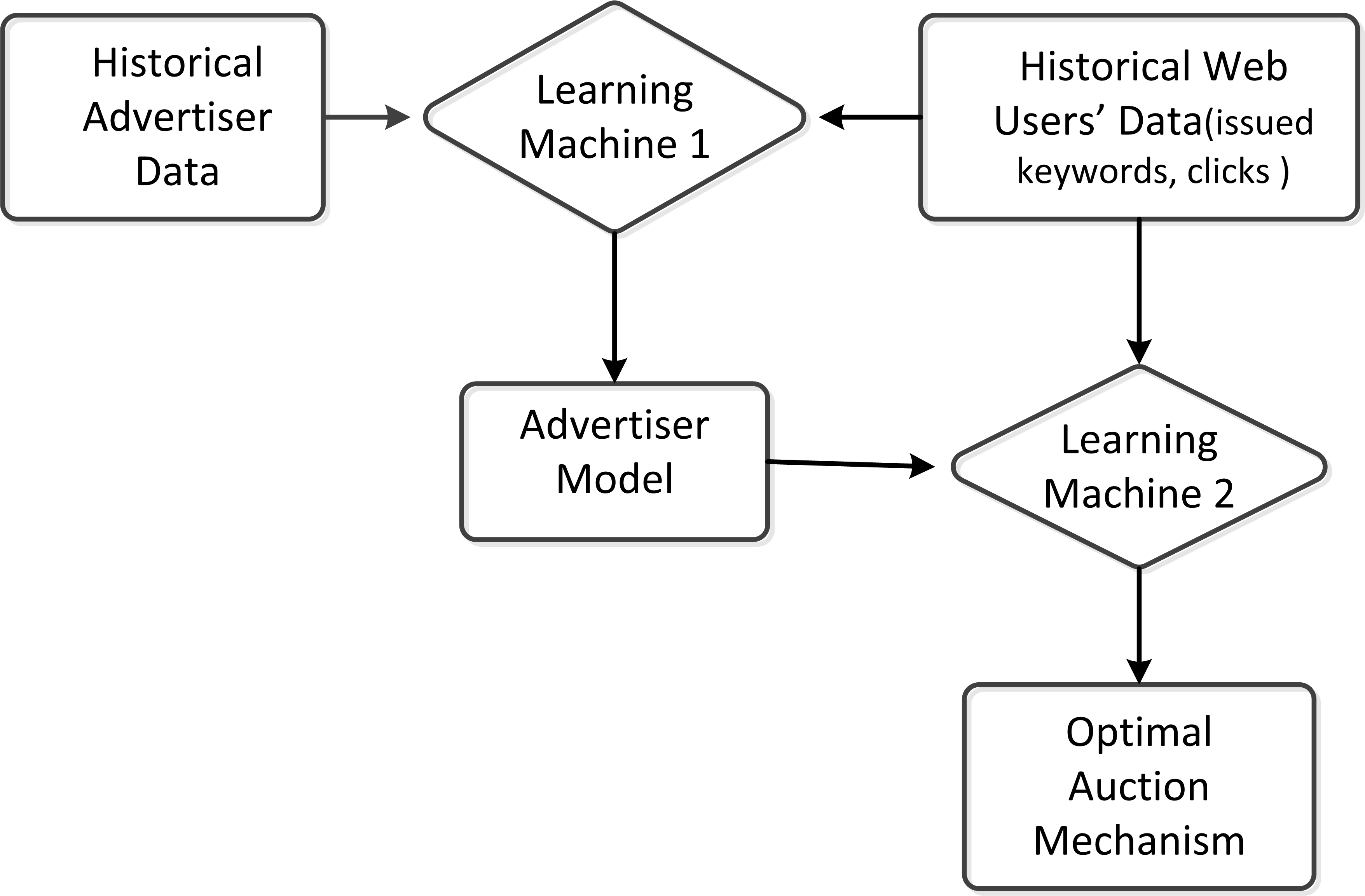}
\caption{Game-Theoretic Machine Learning Framework for Auction Design}\label{fig:framework}
\end{figure}

To realize the proposed framework, first of all, we need to build a reasonable model space $\mathcal{G}$ to describe the advertisers' bidding behaviors. Second, we need to ensure that the revenue computed on the predicted bids by using this model can converge when the prediction period approaches infinity. Only in this way, the subsequent revenue optimization process will make sense. We will make detailed introduction of how we deal with these difficulties in the next sections.

\section{Markov Advertiser Model}


In this section, we show that a Markov process can be used to model how advertisers change their bids. Please note that one can also employ other (more complicated) models in our proposed bilevel optimization framework.

Intuitively, if an advertiser finds that his/her KPI in the last period is below expectation, it is very likely that he/she will increase the bid price in hope to get a better KPI (e.g., ranked in higher positions or receive more clicks). On the other hand, if an advertiser feels that the numbers of impressions and clicks are satisfactory but the cost per click is higher than expected, he/she may lower down the bid price. The Markov property indicates that advertisers' bid changes only depend on their bids in a fixed number of previous time periods. For simplicity here we only consider the time-homogenous first-order Markov property (i.e., the next bid is independent of pervious bids if the current bid is given, and the transition matrix is time-invariant).

To be more specific, we denote the bid vector at time $t$ as $\mathbf{b}^t=(b^t_1, \cdots ,b^t_i, \cdots ,b^t_m)$, in which the $i$-th component $b^t_i$ is the bid of the $i$-th advertiser at time $t$. We assume the bid prices that advertisers can assign come from a finite space $\mathcal{B}$. This assumption is reasonable considering that there is a minimum bid unit in real sponsored search systems and the bid price is always bounded.

Suppose in the $t$-th time period, a stream of users $U^t$ issue the keyword and observe a number of ads. For each user $u\in U^t$, if $u$ clicks on an ad, the corresponding advertiser will pay a certain amount of money to the search engine according to the pricing rule of the auction. At the beginning of the $(t+1)$-th time period, advertiser $i$ will receive a report about the KPI of his/her ad in the $t$-th period. We use $kpi^t_{i}$ to denote this report, which can be considered as the output of a function $KPI_i(\mathbf{b^t}, U^t,f)$. Based on $ kpi^t_{i}$, advertiser $i$ may change his/her bid price to $b^{t+1}_{i}$ in order to be better off. That is,
{\small\begin{eqnarray}
P(b^{t+1}_{i}|\mathbf{b}^{t}, \cdots ,\mathbf{b}^{t},U^1, \cdots ,U^t,f)=P(b^{t+1}_{i}| b_{i}^t, kpi_{i}^t)
\end{eqnarray}}
Denote $M_{i,kpi}$ as the transition matrix for advertiser $i$ with $kpi$.
{\small\begin{equation}
M_{i,kpi}=
\begin{bmatrix}
 \cdots& \cdots& \cdots   \\
 \cdots& p_{i,kpi}(j,k)  & \cdots \\
 \cdots& \cdots& \cdots   \\
\end{bmatrix}
\end{equation}}
where the element $p_{i,kpi}(j,k)$ is the probability of advertiser $i$ changing his/her bid price from $j$ to $k$ given $kpi$, i.e.,
{\small\begin{equation}
p_{i,kpi}(j,k)=P(b^{t+1}_{i}=k|b_{i}^t=j, kpi)
\end{equation}}
Considering that $\mathbf{kpi}^t=( kpi_1^t ,  \cdots ,kpi_m^t)$ is a function of $\mathbf{b^t}$, $U^t$ and auction mechanism $f$, and given $\mathbf{kpi}^t$, advertisers independently change their bids according to the Markov processes, the probability of advertisers' bid vector changing from $\mathbf{b}$ to $\mathbf{b'}$ can be written as,
{\small\begin{eqnarray}
P(\mathbf{b^{t+1}}=\mathbf{b'} | \mathbf{b^{t}}=\mathbf{b}, U^t, f)=\prod_{i} p_{i,KPI_i(\mathbf{b}, U^t,f )}(b_i,b'_i)
\end{eqnarray}}
It is clear that $\mathbf{b^t}$ can also be considered as a Markov process and the transition matrix $Q_{U^t, f}$ is a function of auction mechanism $f$ and user stream $U^t$. We refer to the matrix $Q_{U,f}$ as the \emph{advertiser behavior model}.

For each advertiser $i$ we observe $\{(b_i^t,kpi_i^t,b_i^{t+1})\},t=1, \cdots ,T$ in the historical auction logs and the corresponding transition matrices $M_{i,kpi}$ can be learnt from these observations.  Here we consider two approaches to learn these transition probabilities. The first is a non-parametric approach based on statistical estimation. To be concrete, to estimate $M_{i,KPI}(B_1,B_2)$, i.e. the probability of advertiser $i$ switching his/her bid value from $B_1$ to $B_2$ once his kpi value is $KPI$, one use the term $\hat{M}_{i,KPI}(B_1,B_2)=\frac{\sum_{t=1}^T\mathbbm{1}_{\{kpi_i^t=KPI,b_i^t=B_1,b_i^{t+1}=B_2\}}}{\sum_{t=1}^T\mathbbm{1}_{\{kpi_i^t=KPI,b_i^t=B_1\}}}$. That is, we use empirical frequency to estimate expectation.  Another approach takes a parametric form. In particular, we assume for any $t$, the probability of $b_i^{t+1}$ is proportional to a truncated Gaussian function $p(b)=e^{-(b-\mu)^2}, b \geq 0$, in which the value of $\mu$ depends on $(b_i^t,kpi_i^t)$. We use a linear function to model this dependency: for any advertiser $i$, denote $z_i^t=(b_i^t,kpi_i^t,1)$, then we have

{\small\begin{equation}
\mu = <w, z_i^t>
\end{equation}}
where $w$ can be learnt by maximizing the log likelihood of the observed data using gradient descent(which corresponds to the inner level optimization).
{\small\begin{equation}
w^* = argmin_w (<w , z_i^t> - b^{t+1}_i)^2
\end{equation}}

By using either of the two approaches, the advertiser behavior model can be constructed with the learnt $M_{i,kpi}$, and we can use the model to predict the future bid prices given any auction mechanism. In our experiments, we conducted both non-parametric estimation approach and parametric learning approach, and found that the latter is both more efficient and accurate than the former, thus we will focus on the parametric learning approach in the remaining part of the paper. Once the behavior model is built, tThen the search engine revenue can be computed and optimized on the predicted mechanism-dependent bid prices. Details of this process will be introduced in the next section.

\section{Empirical Revenue Maximization}
In this section, we discuss how to optimize search engine revenue over the bids predicted by the Markov advertiser behavior model. One may have noticed that with the Markov model, the bid prices of the advertisers will always change according to the transition matrix. Then a natural question is whether the empirical revenue on such dynamically changing bid prices will converge when the time period approaches infinity. In the following subsections, we will first give a proof for the convergence, and then propose a Genetic Programming method to optimize the empirical revenue.

\subsection{Convergence Analysis}
First, we formally define the empirical revenue for a given period of time.
\begin{definition}
Given the bid vector $\mathbf{b^t}$ and user stream $U^t$ at each time $t=1, \cdots ,T$, the empirical revenue for the search engine with auction mechanism $f$ is defined as follows,
{\small\begin{equation}
R(\mathbf{b^1}, \cdots ,\mathbf{b^T},U^1, \cdots ,U^T,f)=\frac{1}{T}\sum^T_{t=1}\sum_{u\in U^t}r(\mathbf{b^t},f,u)
\end{equation}}
\end{definition}
Next we show that under certain conditions the empirical revenue will converge in probability. The proof has two components: first we prove the existence of the expectation of the empirical revenue (we call it expected revenue), and then we prove the variance of empirical revenue approaches zero when the time period approaches infinity.
\begin{lemma}
If user stream $U^t$ is stochastically i.i.d. sampled and the expectation of the Markov transition matrix with respect to user stream $\mathbb{E}_UQ_{U,f}$ is positive for auction mechanism $f$, $\lim_{T\to\infty}\mathbb{E}R(b^1, \cdots ,b^T,U^1, \cdots ,U^T,f)$ exists.
\end{lemma}
\begin{proof}
Considering that $U^t$ is independent of $\mathbf{b}^t$, the expected revenue can be formulated as follows,
{\small\begin{eqnarray*}
&&\mathbb{E}R(\mathbf{b^1}, \cdots ,\mathbf{b^T},U^1, \cdots ,U^T,f)\\
&=&\frac{1}{T}\sum_{t=1}^T\sum_{\mathbf{b}\in\{\mathcal{B}\}^{M}}[\mathbb{E}_{U_1, \cdots ,U^{t-1}}P(\mathbf{b^t}=b)\times \mathbb{E}\sum_{u\in U}r(\mathbf{b},f,u)]
\end{eqnarray*}}
It is clear that by taking expectation over $U^1, \cdots ,U^{t-1}$, the bid vector $\mathbf{b}$ follows a Markov process with transition matrix $Q_f = \mathbb{E}_UQ_{U,f}$. Considering that $\mathbb{E}_UQ_{U,f}$ is positive, the Markov process will lead to its stationary distribution $\pi_{Q_f}$, therefore we have,
{\small\begin{eqnarray*}
&&\lim_{T\to\infty} \mathbb{E}R(\mathbf{b^1}, \cdots ,\mathbf{b^T},U^1, \cdots ,U^T,f)\\
&=&\sum_{\mathbf{b}\in\{\mathcal{B}\}^{m}}[\pi_{Q_f}(\mathbf{b})\times\mathbb{E}_U\sum_{u\in U}r(\mathbf{b},f,u)]
\end{eqnarray*}}
where $\pi_{Q_f}(\mathbf{b})$ is the probability of the state $\mathbf{b}$ in the stationary distribution.
\end{proof}

\begin{lemma}
If the user stream $U^t$ is stochastically i.i.d. sampled, and the expected Markov transition matrix $\mathbb{E}_UQ_{U,f}$ is positive for auction mechanism $f$, then $\lim_{T\to\infty}VarR(\mathbf{b}^1, \cdots ,\mathbf{b}^T,U^1, \cdots ,U^T,f) = 0$ .
\end{lemma}
The proof of Lemma 2 is similar to that of Lemma 1.We omit it due to space restrictions. By jointly considering the two lemmas, it is easy to obtain the following theorem.
\begin{theorem}
For $\forall\epsilon>0$, there exists $\delta(\epsilon,T)$ that converges to zero as $T$ approaches infinity, satisfying
{\small\begin{eqnarray}
&&P(|R(\mathbf{b^1}, \cdots ,\mathbf{b^T},U^1, \cdots ,U^T,f)\\
&-&\mathbb{E}R(\mathbf{b^1}, \cdots ,\mathbf{b^T},U^1, \cdots ,U^T,f)|>\epsilon)<\delta(\epsilon,T).\nonumber
\end{eqnarray}}
\end{theorem}

\subsection{Optimization Algorithm}
According to the theoretical guarantee given in the previous subsection, the optimal auction mechanism $f^*\in\mathcal{F}$ can be learnt based on the converged value of the empirical revenue. In this section, we introduce our proposed algorithm to fulfill this task. The algorithm has three components, advertiser behavior learning, empirical revenue simulation, and auction mechanism learning. We call it a bilevel optimization algorithm, or BOA for short. For ease of reference, we summarize our algorithm in Table 1.

Assume that we have advertisers' auction logs and users' query and click-through logs in $T$ time periods for training. The first step is to learn the parametric transition matrices $M^*_{i,kpi}$ based on this training data and construct $Q^*_{U,f}$. As shown in Section 3, this step can be achieved by means of maximum likelihood estimation. Then given the learnt advertiser behavior model, it is a natural idea to select the optimal auction mechanism according to the following expected revenue.
{\small\begin{eqnarray*}
\sum_{\mathbf{b}\in\{\mathcal{B}\}^{m}}[\pi_{Q_f}(\mathbf{b})\times \mathbb{E}_U\sum_{u\in U}r(\mathbf{b},f,u)]
\end{eqnarray*}}
Considering that the empirical revenue will converge to the expected revenue, it suffices to optimize the empirical revenue on a sufficiently long period of time. To achieve this goal, we leverage the learnt advertiser behavior model to predict the bids for another $N$ time periods (please note that this prediction is a part of our learning algorithm, and $N$ can be much larger than $T$). The prediction is performed in the following manner. With an initial bid profile $\mathbf{b}^1$, for any time $1\leq t\leq N$, we randomly sample a piece of user data $\hat{U}^t$ from the training set. Then we are able to rank the ads according to $b^t$ and auction mechanism $f$, check the clicks of $U^t$ on the shown ads based on the historical users click-through logs, and compute $kpi^t_i$ for each advertiser $i$. Based on $kpi_i^t$, we sample advertiser $i$'s next bid $b_i^{t+1}$ according to the learnt transition matrices $M^*_{i,kpi}$. Therefore the empirical revenue of mechanism $f$ over these $N$ time periods can be computed. As shown in Eqn(1), the empirical revenue is a complex function of $f$ due to the second-price formula. We therefore employ the Genetic Programming method for its optimization. Genetic Programming is an evolutionary algorithm developed in the field of artificial intelligence, which can handle complex, non-linear functional relationships.

Here we introduce a method named \emph{$\delta$-sampling technique} to improve the efficiency of the empirical revenue simulation process. Note that to sample the advertisers' bid for each mechanism is infeasible since the mechanism class may be infinite. Therefore, for a mechanism $a$ and a predefined non-negative constant $\delta$, if there exists a mechanism $a^\prime$ satisfying $d(a,a^\prime)\leq\delta$, where $d(\cdot,\cdot)$ is a distance measure between two mechanisms, and the bidding data for $a^\prime$ has been sampled before, we use the empirical bidding data under $a^\prime$ as an surrogate for bidding data under $a$, thus avoiding the sampling process under $a$. By the $\delta$-sampling technique, we greatly improve the efficiency of the proposed optimization algorithm.
\begin{table}[t]
  \caption{The Bilevel Optimization Algorithm}\label{tab:semicrf}
  \centering
{\small
\begin{tabular}{l}
  \hline\hline
  \textbf{Bilevel Optimization Algorithm (BOA)}:\\
  \hline
  \textit{Advertiser Behavior Learning (ABL):}\\
Input : Advertisers' bids $\mathbf{b}^t, t=1, \cdots ,T$.\\
  \quad\quad\quad Advertiser's KPI reports $\mathbf{kpi} ^t, t=1, \cdots ,T$.\\
  1.  Train the advertiser behavior model by maximizing the\\
  \quad log likelihood function in Equ (9).\\
Output : The parametric transition matrix $M^*_{i,kpi}$\\
  \hline
  \textit{Empirical Revenue Simulation (ERS)}\\
Input : The transition matrices $M^*_{i,kpi}$, auction mechanism $f$.\\
  \quad\quad\quad User streams $U^t, t=1, \cdots ,T$\\
  \quad\quad\quad Advertisers' initial bids $\textbf{b}^1$, sampling length $N$\\
\quad\quad\quad A predefined non-negative constant $\delta$\\
  1. For each $t$, uniformly sampling $\hat{U}^t$ from $\{U^t, t=1, \cdots ,T\}$\\
  2. At each time $1\leq t\leq N$,  if there exists a mechanism $f^\prime$ satisfying\\
   \quad $d(f^\prime,f)\leq\delta$, and the bidding data sampling process for $f^\prime$ has been \\
   \quad done before, then use the simulated bidding data under $f^\prime$ as the bidding\\
   \quad data under $f$, otherwise use $\hat{U}^t$, auction mechanism $f$,$\ M^*_{i,kpi}$, and $\mathbf{b}^{t}$ to \\
   \quad predict$\mathbf{b}^{t+1}$.\\
  3. Compute the empirical revenue over the predicted bid sequence. \\
  \textbf{Output}: The empirical revenue of mechanism $f$.\\
  \hline
  \textit{Auction Mechanism Learning (AML)}\\
Input : The transition matrices $M^*_{i,kpi}$.\\
\quad\quad\quad Auction mechanism space $\mathcal{F}$\\
  \quad\quad\quad User streams $U^t, t=1, \cdots ,T$\\
  1. Initialize $K$ random auction mechanisms $f_1, \cdots f_K$.\\
  2. For each auction mechanism $f$ do \\
     \quad 2.1. Compute each $f$'s revenue according to the ERS step. \\
     \quad 2.2. Use $f$'s empirical revenue as the fitness function.\\
  3. New $K$ mechanisms are created by applying genetic operators.\\
  \textbf{Output}: The best mechanism $f^*$.\\\hline\hline
 \end{tabular}}
\end{table}

\section{Experiments}

In this section, we report our experimental results on the effectiveness of our proposed approach. For simplicity and without loss of generality, we assume that the quality score has the mathematical form of $f(x) = (\widehat{CTR}(x))^{\alpha}$. The task of auction mechanism learning then reduces to finding the best parameter $\alpha$. This setting has also been used in some other works, such as \cite{lahaie2007revenue}.

\subsection{Experimental Settings}
Since it is impractical to get the online responses of the advertisers and web users to any learnt auction mechanism, we choose to use a simulation-based experiment. Simulation-based experiments have been widely used in the field of online advertising for algorithm evaluation \cite{li2010exploitation,jordan2010designing}. Generally speaking, we first collect a dataset containing the queries and clicks of web users as well as the bids of advertisers in a period of time. Considering that the clicks and bids in the dataset were generated when a previous auction mechanism was used, we could not use them to evaluate new auction mechanisms due to the second order effect. Therefore, we remove all such information in the dataset and simulate advertisers' bid changes and users' clicks with respect to an underlying auction mechanism using some state-of-the-art behavior models.

Specifically, we use a log data of 300 days obtained from a commercial search engine. We randomly sample ten keywords from the data that contains 1053 advertisers' bid history, and collect all the information about the advertisers and users that have bidden or queried these keywords in the log data. We use the first 100 days to estimate some basic parameters in our experiments (e.g., the CTR and the valuation of each ad) and to construct the training data with the \textit{simulation-based method} by assuming the auction mechanism to be standard GSP ($\alpha=1$ in the quality score).  We then use the remaining 200 days as the test set to compare different mechanisms, and the simulated bids on the test set will be different when different mechanisms are evaluated.

We follow the method described in \cite{richardson2007predicting} to estimate $\widehat{CTR}(x_i)$. We also compute a position discount factor $\beta_j$ by aggregating all the click information. We then simulate users' click behaviors by assuming ad $i$ to have a probability of $\widehat{CTR}(x_i)\beta_j$ to be clicked if it is ranked at position $j$. When estimating the valuation of an ad, we simply use its largest bid in the original dataset.\footnote{Please note that we take this approach for sake of simplicity. One can choose more complicated methods \cite{pin2011stochastic} to estimate the values.} Then we use advertisers' bids in the first day as their initial bids and simulate the bid changes using a mixture model, which is based on three advertiser models proposed in the literature.

1) Best Response Model(BRM) \cite{GSP}. It assumes that each advertiser knows exactly the bids of all his/her competitors in any round of auction. On this basis, he/she will take the best-response bid strategy.

2) Analytical model(AM) \cite{pin2011stochastic}. The model does not assume full information about the bids of advertisers, instead assumes the distribution of the number of advertisers and the distribution of the bid prices of advertisers to be available as public knowledge. On this basis, each advertiser takes the strategy that maximizes his/her expected utility.

3) Stable Behavior Model(SBM). As discussed in the introduction, many advertisers rarely change their bids. This model assumes that advertisers are lazy and never change their bid prices even if the auction mechanism is changed.

In the mixture model, we uniformly sample the coefficients $p1,p2,p3$ from a multinomial distribution, and then select $p_1$ fraction of advertisers to behave according to BRM, $p_2$ fraction of advertisers to behave according to AM, and the rest to behave according to SBM.

We implement three baselines to compare with our proposed approach.

1) The GSP model: the classical GSP auction mechanism with $\alpha=1$ in the quality score $f$.

2) The Worst-Case Analysis (WCA) model: a state-of-the-art optimal mechanism design method \cite{lahaie2007revenue}which studies the revenue in symmetric Nash equilibria and selects the auction mechanism by worst case analysis.

3) The Directly-Learnt Auction (DLA) model: the optimal auction mechanism is directly learnt based on historical log data without considering the second order effect.

When estimating the Markov transition matrix, we run the gradient decent algorithm for 500 iterations. Our proposed BOA algorithm leverages this Markov model and computes the empirical revenue with $N=1000$. When using the Genetic Programming algorithm for empirical revenue optimization, we set the number of individuals in each generation as 10, the number of generations as 50, and crossover/mutation/reproduction rates as 70\%/20\%/10\% respectively. To make our experiments robust, we sample the coefficients in the mixture model for 100 times, and report the average performance in the experimental evaluation.

In the $\delta$-sampling, we set the distance as $d(\alpha_1,\alpha_2)=|\alpha_1-\alpha_2|$ and $\delta$ is set to be 0.01.

\begin{figure}
\centering
\includegraphics[width=0.4\textwidth]{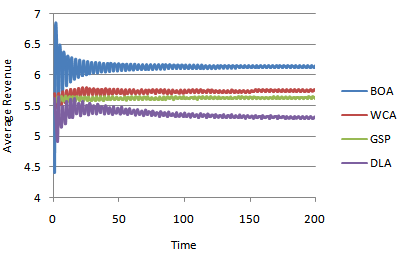}
\caption{Revenue Comparisons}\label{fig:exp}
\end{figure}
\subsection{Experimental Results}
Figure 2 shows the performances of the auction mechanism learnt by our proposed BOA approach as well as the baseline mechanisms. The x-axis corresponds to the length of simulated test data, whose unit is day; and the y-axis corresponds to the average revenue, which is normalized to [0,10] to avoid the disclosure of sensitive business information.

We have the following observations from the figure:

1)	When the size of the test data grows, the average performances of different methods become stable and the comparisons between their performances become clear.

2) The performances of BOA and WCA are better than GSP (with relative improvements of 8.9\% and 2.2\% respectively), which indicates that both our proposed approach and the game-theoretic approach for revenue optimization are effective. Furthermore, BOA is statistically significantly better than WCA and GSP with a p-value=$0.05$, while WCA does not pass the significance test against GSP. This implies that our approach is more effective than previous approaches.

3)	The performance of DLA is worse than GSP (with a relative decrease of 4.8\%), and the result is statistically significant with a p-value=$0.05$. This demonstrates the huge impact of the second-order effect on the experimental results. Due to the strong response from the advertisers to the auction mechanism, simply adopting classical machine learning methods cannot lead to an effective auction mechanism.

From the experiments, we can clearly see the effectiveness of our proposed algorithm and verify our theoretical analysis.

\section{Conclusion and Future Work}

In this paper, we have proposed a \emph{game-theoretic machine learning} approach to deal with the revenue optimization in sponsored search auctions. Specifically, we have proposed a Markov model to describe how advertisers change their bids, and then use the model to learn the auction mechanism that optimizes search engine's revenue on the predicted bids. The experimental results demonstrate the effectiveness of our proposal. As for the future work, we plan to consider other factors in our learning process, e.g.,the reserved price. We also plan to investigate on more comprehensive advertiser behavior models.
\section*{Acknowledgments}
This work is supported by NSFC(61222307, 61075003) and NCET-12-0015, and the work was done when the first author was
visiting Microsoft Research Asia.

\bibliographystyle{named}
\bibliography{ijcai13}

\end{document}